\documentclass{article}
\usepackage{graphicx}
\usepackage{amsmath}
\usepackage{subfigure}

\usepackage[all]{xy}
\usepackage{amsfonts}
\usepackage{amssymb}
\usepackage{amscd}
\usepackage{amsthm}
\usepackage{latexsym}
\usepackage{amsbsy}

\newtheorem{theorem}{Theorem}[section]
\newtheorem{remark}[theorem]{Remark}

\newtheorem{lemma}[theorem]{Lemma}

\newtheorem{Note}[theorem]{Note}

\newcommand{\rnc}[2]{\renewcommand{#1}{#2}}
\rnc{\theequation}{\thesection.\arabic{equation}}
\rnc{\[}{\begin{equation}}
\rnc{\]}{\end{equation}}

\begin{document}

\title{A note on "eight-vertex" universal quantum gates}
\author{Arash Pourkia\\Department of Mathematics, \\ American University of the Middle East (AUM), Egaila, Kuwait}

\maketitle

\begin{abstract}
Many well-known and well-studied four by four universal quantum logic gates in the literature are of a specific form, the so called eight-vertex form \eqref{8vertexform} \cite{kaufman etal 05-1,kaufman etal 05-2}, or {\it similar} to it. 

We present a formalism for universal quantum logic gates of such a form. First, we provide explicit formulas in terms of matrix entries, which are the necessary and sufficient conditions for such a matrix to be a solution to the Yang-Baxter equation \eqref{byb}. Then, combining this with the conditions needed for being unitary \eqref{unitarycond} and being entangling \eqref{entanglingcond}, we give a full description of entangling unitary solutions to the Yang-Baxter equation (hence, universal quantum logic gates) of such a specific form. We investigate in detail all the possible cases where some of the eight main entries might or might not be zero. 
\end{abstract}

\section{Introduction}
Kauffman and Lomonaco \cite{kaufman lomonaco 02, kaufman lomonaco 03, kaufman lomonaco 04} have pioneered the recent studies of relations between quantum entanglement and topological entanglement, in the field of quantum computing. In particular, they have suggested that the {\bf unitary solutions to the Yang-Baxter equation (Y-B) \eqref{byb}}  are crucially important in exploring such relations. Because on one hand, they are a source of universal quantum logic gates. On the other hand, they could yield some invariants of links/knots. The studies in this direction have been continued and further expanded by many authors, e.g., in \cite{kaufman etal 05-1,kaufman etal 05-2, yzhang,dye 03, arash-josep-2016,kzhang Yonzhang,jchen et al 07,pinto et al 13,alagic et al 15,lho et al 10}.\newline

When it comes to the four by four unitary solutions to the Yang-Baxter equation (Y-B) (which are universal quantum logic gates), the so called eight-vertex form \cite{kaufman etal 05-1,kaufman etal 05-2}
\[ \begin{pmatrix}
*&0&0&* \\
0&*&*&0 \\
0&*&*&0 \\
*&0&0&*
\end{pmatrix} \label{8vertexform}\]
is ubiquitous in the literature. They have been extensively studied in \cite{kaufman etal 05-1,kaufman etal 05-2}. All four by four unitary solutions to the Yang-Baxter (Y-B) equation classified in \cite{dye 03}, based on \cite{hietarinta}, are either of the form or {\it{similar}} to a form of \eqref{8vertexform}.  In \cite{arash-josep-2016} we obtained an infinite family of universal quantum logic gates, in the form \eqref{8vertexform}, related to cyclic groups of order $n$. More famously, as a special case, the Bell matrix \cite{kaufman etal 05-1,kaufman etal 05-2, yzhang} is of this particular form. More examples are included but not limited to the ones studied in \cite{kaufman lomonaco 02,kaufman lomonaco 03, kaufman lomonaco 04, kaufman etal 05-1,kaufman etal 05-2, yzhang, dye 03, arash-josep-2016, kzhang Yonzhang}. \newline

In the present paper we provide a general formalism for four by four entangling unitary solutions to the Yang-Baxter equation of the form \eqref{8vertexform}. We provide explicit formulas in terms of matrix entries, which are the necessary and sufficient conditions for such a matrix to be a solution to the Y-B equation  \eqref{byb}. We Combine this with the conditions needed for being unitary \eqref{unitarycond} and being entangling \eqref{entanglingcond}  to give a full polynomial description in terms of matrix entries of universal quantum logic gates of the form \eqref{8vertexform}. We investigate in detail all the possible cases where some of the eight main entries might or might not be zero. \newline

This paper is organized as follows. In Section \ref{preliminaries} we provide, very briefly, the main preliminaries needed for the succeeding sections. 
In Section \ref{mainlemma}, we give a detailed description, in terms of entries, for a matrix of the form,
$R=\begin{pmatrix}
a&0&0&b \\
0&c&d&0 \\
0&e&f&0 \\
g&0&0&h
\end{pmatrix}$ to be a unitary solution to Y-B equation \eqref{byb}, hence a quantum logic gate. Then we analyze all the possible cases where some of the entries might be zero or non of them zero. Finally for each case we investigate the conditions on entries that make $R$ entangling, and hence universal as a quantum logic gate \cite{brylinskis 02}. In Section \ref{summary} we give a concise summary of our results from previous sections.

\section{Preliminaries} \label{preliminaries}
Let $V$ be a $n$ dimensional complex vector (Hilbert) space, and let $R:V\otimes V \rightarrow V\otimes V$ be a linear map. $R$ can be represented as a $n^2$ by $n^2$ matrix for some basis of $V$. $R$ is said to be a solution to the (braided) Y-B equation, if it satisfies the following relation, \cite{kaufman lomonaco 02, kaufman lomonaco 04}.
\[ (R\otimes I)(I\otimes R)(R\otimes I)=(I\otimes R)(R\otimes I)(I\otimes R), \label{byb}\]
where $I$ is the identity map on $V$. 
\begin{remark}
Follwoing \cite{kaufman lomonaco 04} we call the above equation the braided Y-B equation. It is true that composing $R$ with the swap gate,
\[ S=\begin{pmatrix}
1&0&0&0 \\
0&0&1&0 \\
0&1&0&0 \\
0&0&0&1
\end{pmatrix} \label{swapgate},\] results in a solution to the so called algebraic Y-B equation (and vice versa), 
\[R_{12}R_{13}R_{23}=R_{23}R_{13}R_{12}, \label{ayb}\]
 where $R=\sum_{i} \mathfrak{r}_i\otimes \mathfrak{s}_i$, $R_{12}=\sum_{i} \mathfrak{r}_i\otimes \mathfrak{s}_i\otimes 1$, $R_{13}=\sum_{i} \mathfrak{r}_i\otimes 1\otimes \mathfrak{s}_i$, and $R_{23}=\sum_{i} 1\otimes \mathfrak{r}_i\otimes \mathfrak{s}_i$, \cite{kaufman lomonaco 04}.

\end{remark}
The topic of topological entanglement in a nutshell is as follows \cite{kaufman lomonaco 02, kaufman lomonaco 03, kaufman lomonaco 04}. If $R$ is invertible, it provides an infinite family of braid group representations. Which, in  turn, yield some invariants of links/knots. In such a process $R$ (or $R^{-1}$) is the operator for over-crossing (or under-crossing) moves \cite{jones1985, kauffman1987, turaev 1992}.\newline

As for the topic of quantum entanglement in quantum computing, it could be summarized as follows. When $R$  is unitary, i.e.,
\[ R^{-1}=R^{\dagger}, \label{unitarycond}\]
where $R^{\dagger}$ is the conjugate transpose of $R$, then $R$ could be considered as a quantum logic gate acting on $n$-qudits $|\Phi \rangle$, which are $n^2$-component vectors in $V \otimes V$. A {\it quantum state} which, abstractly, is represented by a $n$-qudit $|\Phi \rangle$ is called to be entangled if it can not be decomposed into tensor product of two other states from $V$. $R$ is said to be an entangling quantum gate if it creates entangled states by acting on un-entangled ones \cite{kaufman lomonaco 02, kaufman lomonaco 03, kaufman lomonaco 04, brylinskis 02}. In other words: 
 \[ {\text {$R$ is entangling if there exist an un-entangled state}} \, |\Phi \rangle \, {\text {such that}}\, R|\Phi \rangle \, {\text {is entangled}} \label{entanglingcond}\] 
A quantum gate $R$ is {\it universal} if and only if it is entangling \cite{brylinskis 02}.\newline

In the present paper we focus only on two dimensional complex vector (Hilbert) spaces $V$. Thus, $R:V\otimes V \rightarrow V\otimes V$ can be represented by a $4$ by $4$ matrix with entries in $\mathbb{C}$, with respect to a basis of $V$, namely $\{ |0\rangle , |1\rangle \}$. 
For $V\otimes V$, we preferably use the so called computational basis for two qubit states: 
$\{ |00\rangle ,|01\rangle,|10\rangle,|11\rangle \}$.

\section{Main Lemma} \label{mainlemma}
In this section, in the following lemma we provide the necessary and sufficient conditions in terms of matrix entries, for a matrix of the form,
$R=\begin{pmatrix}
a&0&0&b \\
0&c&d&0 \\
0&e&f&0 \\
g&0&0&h
\end{pmatrix}$ to be a solution to Y-B equation \eqref{byb}. Plus the conditions for $R$ to be unitary. In the follwing sub-sections we will investigate what happens to the conditions on entries if we also want some of the entries to be zero or non of them zero. There, we also investigate the entangling (thus the universality) property. 

\begin{lemma} \label{lemma0}
Let $R$ be the following $4$ by $4$ matrix,
 
\begin{align}\label{generalR}
R&=
\begin{pmatrix}
a&0&0&b \\
0&c&d&0 \\
0&e&f&0 \\
g&0&0&h
\end{pmatrix}\\\nonumber
\end{align}
where, $a,b,c,d,e,f,g,h$ are complex numbers.

\begin{item}
\item[a)] $R$ is a solution to Y-B equation \eqref{byb}, if and only if:
\[bfg=bcg \label{yb1}\]
\[abc+be^2=a^2b+bch \label{yb2}\]
\[a^2c+bhg=ac^2+dce \label{yb3}\]
\[bge=dcf \label{yb4}\]
\[fce=bdg \label{yb5}\]
\[cf^2=c^2f \label{yb6}\]
\[ga^2+hcg=cga+d^2g \label{yb7}\]
\[gab+ch^2=c^2h+dce \label{yb8}\]

\[abd+bef=abe+bdc \label{yb9}\]
\[a^2b+bfh=abf+bd^2 \label{yb10}\]
\[bgf=cbg \label{yb11}\]
\[acb+bh^2=cbh+d^2b \label{yb12}\]
\[cfd=ebg \label{yb13}\]
\[ceb+dbh=ebh+fdb \label{yb14}\]
\[dcf=gbe \label{yb15}\]
\[cgb=gbf \label{yb16}\]

\[fbg=bgc \label{yb17}\]
\[efc=bgd \label{yb18}\]
\[ega+fdg=ceg+dga \label{yb19}\]
\[fce=dgb \label{yb20}\]
\[ga^2+hfg=e^2g+fga \label{yb21}\]
\[gbc=fgb \label{yb22}\]
\[gca+h^2g=ge^2+hgc \label{yb23}\]
\[gdc+hge=gef+hgd \label{yb24}\]

\[efd+f^2a=a^2f+bhg \label{yb25}\]
\[e^2b+fbh=afb+bh^2 \label{yb26}\]
\[f^2c=fc^2 \label{yb27}\]
\[egb=cfd \label{yb28}\]
\[gbd=efc \label{yb29}\]
\[gd^2+hgf=gfa+h^2g \label{yb30}\]
\[gab+h^2f=efd+hf^2 \label{yb31}\]
\[gcb=gfb \label{yb32}\]

\item[b)] $R$ is unitary (i.e.  $R^{-1}=R^{\dagger}$, where $R^{\dagger}$ is the conjugate transpose of $R$), if and only if,
 \[a\bar{a} + b\bar{b}=1,\quad c\bar{c} + d\bar{d}=1,\quad e\bar{e} + f\bar{f}=1,\quad g\bar{g} + h\bar{h}=1  \label{unitary1}\] and 
\[a\bar{g} + b\bar{h}=0,\quad c\bar{e} + d\bar{f}=0  \label{unitary2}\] 
\end{item}
\end{lemma}

\begin{proof}
Proof is by direct and straightforward calculations from the definitions, with the proof of part (a) being lengthy.  
\end{proof}

\begin{Note} \label{note1}
It is important to notice the followings. As results of relations \eqref{unitary1} and \eqref{unitary2} in Lemma \ref{lemma0} we have, on one hand,
\[b=0 \Leftrightarrow g=0,\quad c=0 \Leftrightarrow f=0, \quad e=0 \Leftrightarrow d=0, \quad a=0 \Leftrightarrow h=0 \label{note1}   \]

On the other hand, 
\[b=0 \Rightarrow a\neq 0, \quad c=0 \Rightarrow d\neq 0, \quad e=0 \Rightarrow f\neq 0, \quad g=0 \Rightarrow h\neq 0, \label{note2}   \]
In what follows, we will use the above facts very often, without necessarily referring to them.  
\end{Note}

\begin{Note} \label{note2}
It also important to notice that all the results of the above lemma and the rest of results that follows in subsequent sections will stay valid, if we replace $R$ with any phase deformation of $R$ by a phase factor $e^{i\varphi}$. Namely, replacing $b$ with $be^{i\varphi}$ and $g$ with $ge^{-i\varphi}$ in all the above and in all what follows, will not effect any of our results.
\end{Note}

{\bf {When some or non of entries are zero: }} 
In the following sub-sections, for $R$ as a unitary solution to YB equation, we will investigate the conditions on entries when some of them or non of them are zero. In each case we will find conditions that guarantee the gate $R$ to be entangling and therefore universal \cite{brylinskis 02}. We start with all the possible cases when some entries of $R$ are zero. \newline

\subsection{\bf{Case 1: $b=g=0$}} \label{b=g=0}
When $b=g=0$ (which, from Note \ref{note1}, implies $a, h\neq 0$), from \eqref{yb4} we have, $dcf=0$. Therefore we can have only one of two cases, consistent with all the relations \eqref{yb1}-\eqref{unitary2}:

{{\bf{Either:}} $c=f=0$ (which implies $d, e\neq 0$), for which we have: 
\begin{align}\label{case1-1}
R&=
\begin{pmatrix}
a&0&0&0 \\
0&0&d&0 \\
0&e&0&0 \\
0&0&0&h
\end{pmatrix}\\\nonumber
\end{align}
where, $|a|=|d|=|e|=|h|=1$.\\ 

In this case $R$ is entangling \cite{kaufman lomonaco 04}, and therefore a universal quantum logic gate  \cite{brylinskis 02}, if $ah \neq de$. To see this it is enough, for example, to act $R$ on the un-entangled state 
$\begin{pmatrix}
1 \\
1 \\
1 \\
1
\end{pmatrix}\nonumber$ in the computational basis.

In fact the value $|ah-de|$ is a measure of entanglement. Because if we apply $R$ to an arbitrary un-entangled state
$\begin{pmatrix}
x \\
y \\
z \\
w
\end{pmatrix}\nonumber$, where, $xw=yz \neq 0$, the outcome
$\begin{pmatrix}
ax \\
dz \\
ey \\
hw
\end{pmatrix}\nonumber$ is entangled \cite{kaufman lomonaco 02,kaufman lomonaco 04}, only if ($xw \neq 0$ and) $ah-de \neq 0$. More rigorously, the concurrence \cite{wooter-hill-1997} of the outcome state is $2|(ah-de)xw|$.\\

{\bf{Or:}} $d=e=0$ (which implies $c, f\neq 0$), for which we have the followings: 
From \eqref{yb3} $a=c$, from \eqref{yb6} $c=f$, and from \eqref{yb8} $c=h$. Therefore:
\begin{align}\label{case1-2}
R&=
\begin{pmatrix}
a&0&0&0 \\
0&a&0&0 \\
0&0&a&0 \\
0&0&0&a
\end{pmatrix}=aI\\\nonumber
\end{align}
where, $|a|=1$. In this case, obviously $R$ is not entangling. 

\subsection{\bf{Case 2: $c=f=0$}} \label{c=f=0}
When $c=f=0$, from \eqref{yb3} we have, $bhg=0$. Therefore there are two cases to discuss:

{\bf{Either:}} $b=g=0$. But this is just the case \eqref{case1-1} discussed in {\bf {Case 1}} above.

{\bf{Or:}}  $a=h=0$, for which we have from \eqref{yb2} $e=0$ (which implies also $d=0$). But this contradicts the fact that $c=f=0$ (refer to Note \eqref{note1}). 

\subsection{\bf{Case 3: $a=h=0$}} \label{a=h=0}
When $a=h=0$ (which implies $b,g \neq 0$), From \eqref{yb1} we have $c=f$. Also from \eqref{yb2} we have $e=0$ (therefoer $d=0, and $ and $c,f \neq 0$). All these mean: 
\begin{align}\label{case3-1}
R&=
\begin{pmatrix}
0&0&0&b \\
0&c&0&0 \\
0&0&c&0 \\
g&0&0&0
\end{pmatrix}\\\nonumber
\end{align}
where, $|b|=|c|=|g|=1$.\\ 

Similar to the case $c=f=0$ in Case \eqref{b=g=0}, in this case $R$ is entangling \cite{kaufman lomonaco 04} and therefore a universal quantum logic gate  \cite{brylinskis 02}, if $bg \neq c^2$. In fact, $|bg-c^2|$ is a measure of entanglement. Because, if we apply $R$ to an arbitrary un-entangled state
$\begin{pmatrix}
x \\
y \\
z \\
w
\end{pmatrix}\nonumber$, where, $xw=yz \neq 0$, the outcome
$\begin{pmatrix}
bw \\
cy \\
cz \\
gx
\end{pmatrix}\nonumber$ is entangled \cite{kaufman lomonaco 02,kaufman lomonaco 04}, only if ($xw \neq 0$ and) $bg-c^2 \neq 0$. More rigorously, the concurrence \cite{wooter-hill-1997} of the outcome state is $2|(bg-c^2)xw|$.\\

\subsection{\bf{Case 4: $d=e=0$}} \label{d=e=0}
When $d=e=0$ (which implies $c,f \neq 0$), first we realize that if at the same time $b=g=0$ then we arrive at the case \eqref{case1-2} discussed in {\bf {Case 1}} above. Also if $a=h=0$ then we arrive at the {\bf{case 3}} discussed above.

Therefore we assume $b,g \neq 0$ and $a,h \neq 0$. Then we will have the followings, consistent with the relations \eqref{yb1}-\eqref{yb32}:\\

From \eqref{yb1}, $c=f$. From \eqref{yb2}, $ac=a^2+ch$. From \eqref{yb3}, $a^2c+bhg=ac^2$. From \eqref{yb8}, $gab+ch^2=c^2h$. From \eqref{yb12}, $ac+h^2=ch$.  \\

Comparing the second and the last of these relations we get $h^2=-a^2$ or $h=\pm ia$. Putting this back into the second relation gives us $a=(1\mp i)c$. But now considering the unitary relations \eqref{unitary1} and \eqref{unitary2} implies that $|a|=\sqrt{2}\,|c|=\sqrt{2}$, which contradicts the unitary relations!. \newline

\subsection{No zero entries \label{nozeroentries}}
Now we proceed to the case when no entries of $R$ are zero. First from \eqref{yb1} we have:
\[c=f \label{c=f}\]
Next, from \eqref{yb10} and \eqref{yb12} we have:
\[a^2-d^2=d^2-h^2=ac-hc \label{a^2-d^2=d^2-h^2=ac-hc}\]
Which in turn implies: 
\[a^2+h^2=2d^2 \label{a^2+h^2=2d^2} \]

Next, from \eqref{yb2} and \eqref{yb7}we have: 
\[ d^2=e^2 \quad \Longleftrightarrow \quad e=d \quad OR \quad e=-d \label{d2=e2}\]
From here, we analyze all possible outcomes of the relations \eqref{yb1}-\eqref{yb32} in two cases, either $e=d$ or $e=-d$.\newline

{\bf Either $e=-d$:} From \eqref{yb9} and \eqref{yb14}, respectively, we have: 
\[a=c \quad {\text{and}} \quad c=h \label{a=c=h}\]
These together with \eqref{a^2+h^2=2d^2} imply:
\[d^2=a^2 \quad \Longleftrightarrow \quad a=\pm d \label{d^2=a^2}\]
Also from \eqref{yb13} we have:
\[bg=-c^2 \label{bg=-c^2}\]

Therefore $R$ has the following form:
\begin{align}\label{case5-1}
R&=
\begin{pmatrix}
a&0&0&b \\
0&a&\pm a&0 \\
0&\mp a&a&0 \\
g&0&0&a
\end{pmatrix}\\\nonumber
&=
a\begin{pmatrix}
1&0&0&p \\
0&1&\pm 1&0 \\
0&\mp 1&1&0 \\
-p^{-1}&0&0&1
\end{pmatrix}\\\nonumber
\end{align}
where, $bg=-a^2$. Also from relations \eqref{unitary1} and \eqref{unitary2} we have, $|a|=|b|=|g|={\dfrac{1}{\sqrt{2}}}$. In the second representation, $p=\dfrac{b}{a}$.\newline

{\bf Or $e=d$:} As in \eqref{c=f},  \eqref{a^2-d^2=d^2-h^2=ac-hc}, and \eqref{a^2+h^2=2d^2} we have: 
\[ {c=f}, \quad {a^2-d^2=d^2-h^2=ac-hc}, \quad {a^2+h^2=2d^2} \label{e=drlations}\] 
Also from \eqref{yb13} we have:
\[bg=c^2 \label{bg=c^2}\]

Therefore $R$ has the following form:
\begin{align}\label{case5-2}
R&=
\begin{pmatrix} 
a&0&0&b \\
0&c&d&0 \\
0&d&c&0 \\
g&0&0&h
\end{pmatrix}\\\nonumber
\end{align}
where, $a,b,c,d,g,h$ satisfy the relations \eqref{e=drlations} and \eqref{bg=c^2},  plus the unitary relations \eqref{unitary1} and \eqref{unitary2}.

\begin{remark}
We remark that $R$ with no zero entries is always an entangling and therefore a universal quantum gate \cite{kaufman lomonaco 04, brylinskis 02}. This is because, for example, 
$R\begin{pmatrix}
1 \\
0 \\
0 \\
0
\end{pmatrix}=\begin{pmatrix}
a \\
0 \\
0 \\
g
\end{pmatrix}$, which is obviously entangled, since $ag \neq 0$. 
\end{remark}

\begin{remark}
We also remark that the quantum gates introduced in \cite{arash-josep-2016} are a very special case of $R$ \eqref{case5-2} in the last case, when $h=d=a$, $g=b=-c$, satisfying unitary relations \eqref{unitary1}.
\end{remark}

\section{Summary} \label{summary}
Here we summarize, for a matrix $R$ of the form \eqref{8vertexform}, all the possible cases where $R$ is a unitary solution to Y-B equation \eqref{byb} and also entangling. Therefore a universal quantum logic gate.

{\bf 1:}
\begin{align}\label{case1-s}
R&=
\begin{pmatrix}
a&0&0&0 \\
0&0&d&0 \\
0&e&0&0 \\
0&0&0&h
\end{pmatrix}\\\nonumber
\end{align}
where, $|a|=|d|=|e|=|h|=1$. $R$ is entangling and therefore a universal quantum logic gate iff $ah \neq de$. \\

{\bf 2:}
\begin{align}\label{case2-s}
R&=
\begin{pmatrix}
0&0&0&b \\
0&c&0&0 \\
0&0&c&0 \\
g&0&0&0
\end{pmatrix}\\\nonumber
\end{align}
where, $|b|=|c|=|g|=1$. $R$ is entangling and therefore a universal quantum logic gate, iff $bg \neq c^2$. \\

{\bf 3:}
\begin{align}\label{case3-s}
R&=
\begin{pmatrix}
a&0&0&b \\
0&a&\pm a&0 \\
0&\mp a&a&0 \\
g&0&0&a
\end{pmatrix}\\\nonumber
\end{align}
where, $bg=-a^2$, and $|a|=|b|=|g|={\dfrac{1}{\sqrt{2}}}$. $R$ is always entangling and therefore a universal quantum logic gate. \newline

{\bf 4:}
\begin{align}\label{case4-s}
R&=
\begin{pmatrix} 
a&0&0&b \\
0&c&d&0 \\
0&d&c&0 \\
g&0&0&h
\end{pmatrix}\\\nonumber
\end{align}
where, $a,b,c,d,g,h$, non of them zero, satisfy the relations,
\[  {a^2-d^2=d^2-h^2=ac-hc}, \quad {a^2+h^2=2d^2}, \quad bg=c^2 \label{4} \] 
  plus the unitary relations,  
 \[a\bar{a} + b\bar{b}=1,\quad c\bar{c} + d\bar{d}=1,\quad g\bar{g} + h\bar{h}=1  \label{unitary1r}\] and 
\[a\bar{g} + b\bar{h}=0,\quad c\bar{d} + d\bar{c}=0  \label{unitary2r}\]

$R$ is always entangling and therefore a universal quantum logic gate.\newline

The author would like to thank Dr. Josep Batle, the co-author in \cite{arash-josep-2016}, for introducing him to this field of research.

\end{document}